\documentclass[11pt]{article}
\usepackage{amsmath}
\usepackage{amsfonts}
\usepackage{amssymb}
\usepackage{subfig}
\usepackage{graphicx}

\usepackage{xcolor}

\usepackage{amsthm}

\newtheorem{theorem}{Theorem}

\newtheorem{corollary}{Corollary}

\title{A Random Dot Product Model for Weighted Networks}

\author{Daryl R. DeFord$^{1}$ \and\!\!\!\!\!\!\! Daniel N. Rockmore$^{1,2,3}$
\\
\normalsize{$^{1}$Department of Mathematics, Dartmouth College, Hanover, NH, USA 03755}\\
\normalsize{$^{2}$Department of  Computer Science Dartmouth College, Hanover, NH, USA 03755}\\
\normalsize{$^{3}$The Santa Fe Institute, Santa Fe, NM, USA 87501}}

\usepackage{hyperref}

\begin{document}
\maketitle

\begin{abstract}
This paper presents a generalization of the random dot product model for networks whose edge weights are drawn from  a parametrized probability distribution. We focus on the case of integer weight edges and show that many previously studied models can be recovered as special cases of this generalization. Our model also determines a dimension--reducing embedding process that gives geometric interpretations of community structure and centrality. The dimension of the embedding has consequences for the derived community structure and we exhibit a stress function for determining  appropriate dimensions. We use this approach to analyze a coauthorship network and voting data from the U.S. Senate. 

\end{abstract}


\section{Introduction}
\subsection{Background}

Since the introduction of random graph models by Erdos and Renyi \cite{er1}, and Gilbert \cite{G1} in 1959, generative graph models have become a fundamental tool for understanding the statistical properties of complex networks \cite{G2,P0}. While the so-called ``E-R networks" provide a null model that match (by design) the basic statistic of average degree, it preserves no other important structural properties, thereby making it something of a strawman for comparison. The same is true even for the   more sophisticated   configuration model (preserving degree sequence). Other generative models, such as the Barabasi--Albert preferential  attachment model \cite{ba1} and the Watts--Strogatz small world model \cite{ws1}, have been introduced to more closely approximate  the behavior of observed networks.

The {\em stochastic block mode}l (SBM) is a  generative network process that has of late attracted a great deal of attention \cite {A1,G2} for its   ability to incorporate community structure.   It starts with the assumption of community structure encoded in terms of intra- and inter-community network probabilities \cite{H1,K0}. See Section 3.8 of \cite{G2} and the introduction to \cite{C_1} for a good review as well as applications. It is of particular interest in the social sciences where it is commonly used as a framework for community detection \cite{F0,K0,P0}.  Statistical approaches to the  determination of model parameters for individual networks is an active area of study \cite{A0,A1,C_M1}. 

The work cited above generally relates to simple unweighted networks. Constructing generative models for weighted networks to complement these is necessary to fully understand intrinsically weighted systems \cite{C_1}. Simply transforming the network of interest into an unweighted network is often not useful and at worst is confounding. For example simple thresholding can and will generally obscure properties of the network (see e.g.,  \cite{T1}).

Standard null models in this area may permute the edge weights (see e.g., \cite{P1}) or randomly reassign according to some distributional considerations (see e.g.,  \cite{B3,L1}). Generalizing tools that have been developed for simple networks to weighted or multilayer networks is an active and important area of study \cite{D1,K,N0}.  In \cite{C_1,C} a weighted stochastic block model is presented, along with examples demonstrating community structures that are obscured by thresholding. They also use a weighted approach to analyze several real--world weighted networks, including airport flights (weighted by number of passengers) \cite{C0}, international scientific collaboration (weighted by number of papers that include authors from both nations) \cite{P}, and congressional voting (weighted by a normalized ``interlock'' measure) \cite{PP}, among others. 

In this paper, we introduce the {\em Weighted Random Dot Product Model} (WRPDM), a new kind of generative process for weighted networks based on  a generalization of the {\em Random Dot Product Model} (RDPM).  The RDPM generalizes many versions of the SBM \cite{O} and has proved useful for proving statistical results about SBM community detection \cite{C_M1,C_M2}.  We focus our attention on networks with integer edges weights, but the methodology presented   is applicable to more general weighting systems.  We also apply this new formulation to several examples.  

\subsection{Random Dot Product Models} The RDPM is a latent space model introduced by Kraetzel et al. \cite{K2,N3} and further developed by Young and Scheinerman \cite{S_Y2,S_Y1}. The RDPM is a generalization of dot product representation graphs \cite{F1}, which are generalizations of interval graphs and intersection graphs, of interest for decades as combinatorial objects  \cite{H2,K1}. In the RDPM process, a dimension $d$ is chosen and each node in the network is associated  to a vector in $\mathbb{R}^d$, drawn from a fixed distribution over $\mathbb{R}^d$. From this, edges are determined according to Bernoulli trials with the probability of an edge occurring between any two nodes  given by the (suitably constrained) dot product of the two associated vectors. In the original papers, analytical results  on the expected properties of networks drawn from the RDPM were derived for the cases of $d=1$ and $d=2$ \cite{K2,N3}.  Results for larger $d$ were studied in \cite{S_Y2,S_Y1}. Values of $d$ up to 24 were used in numerical experiments in \cite{O} while  \cite{M0} compared small values of $d$ for a data set of international relations.

The RDPM process also motivates an adjacency spectral embedding for networks drawn from a SBM. For a given network, the adjacency embedding \cite{C_M3,C_M1}, assigns the corresponding row vector from a normalized spectral decomposition of the adjacency matrix to each node, approximating the RDPM for the network (e.g. see Lemma 2.5 of \cite{C_M3}).  The linear algebraic properties of this embedding allow for asymptotic statistical results to be proved about the community assignment formulation of the SBM \cite{C_M3,C_M2}. This includes a use for proving statistical consistency results about community assignment, including a hypothesis testing procedure over the distribution of original latent positions \cite{C_M3,C_M1,C_M2}.   The hypothesis test detects whether two observed graphs were drawn from the same underlying RDPM or SBM process \cite{C_M2}.

More formally, following the standard definitions as  in \cite{S_Y2}, we describe the RDPM generative model as a sequence of steps:
\begin{enumerate}
\item[]{\bf (RDPM 1)}: Select the number of nodes $n$.
\item[]{\bf (RDPM 2)}:  Select the latent dimension $d$.
\item[]{\bf (RDPM 3)}:  Select a distribution $W$ over $\mathbb{R}^d$ with $\mathbb{P}\left(\langle X, Y\rangle\in (0,1)\right)=1$ where $X$ and $Y$ are drawn independently from $W$.
\item[]{\bf (RDPM 4)}:  For each node, $1\leq j\leq n$, draw a vector, $X_j\in\mathbb{R}^d$ from $W$.
\item[]{\bf (RDPM 5)}:  Form an adjacency matrix, $A$, form a network with  $A_{j,\ell}$ drawn from $\operatorname{Bernoulli}(\langle X_j, X_\ell\rangle)$ for $j\neq \ell$ and $A_{j,j}=0$ for all $1\leq j\leq n$. 
\end{enumerate}

{\em \bf Notation.} Throughout this paper we will use this notation consistently with $n$ representing the number of nodes in a given network, $d$ being the latent dimension, and nodes indexed by $j$ and $\ell$. When we consider probability distributions with multiple parameters we will use $k$ to describe the number of parameters and index the parameters with $i$. 

Clearly, the choice of distribution $W$ greatly influences the expected statistics of networks generated with the RDPM. However, it is still possible to prove bounds on various expected network statistics, even for arbitrary $W$ \cite{S_Y2,S_Y1}. In particular, these papers show  that for a broad class of distributions $W$, networks generated with the RDPM exhibit clustering and short average path lengths as would be expected in small world networks. In some cases it can also be shown that the expected degree distribution follows a power law \cite{S_Y2,S_Y1}. In Section 2  we will show how restrictions of the RDPM can describe other commonly studied generative models including the Erdos--Renyi model, stochastic block models, and the Chung--Lu model.

 Another key feature of the RDPM is the interpretability of the vectors assigned to each node. Like many latent space models, two nodes whose vectors  are close together in $\mathbb{R}^d$ are more likely to be connected in the corresponding network. However,  the use of vectors and a dot product instead of simply points and a distance measure allows us to understand the propensity to connect along two different axes: ``similarity'' (as captured by the angle between the vectors) and ``significance'' (as captured by the magnitudes of the involved vectors) \cite{S_T,S_Y2}. Geometric interpretations of similarity are common to many dimension reduction and latent space models. In the RDPM, these similarities occur because the closer two vectors are to pointing in the same direction, the larger their dot product will be. However, magnitude of the vector also contributes to the size of the dot product ($\langle u,v\rangle=||u||\cdot||v|| \cos(u,v)$), so for a fixed direction, an increase in magnitude would contribute to an increase in its dot product with any other vector, thus increasing its propensity to form edges (as per the RDPM). This is what is meant by significance.  As we will discuss in Section 4.2 and Section 4.3 these two factors can be thought of as representing the group membership (angle--similarity) and centrality (magnitude--significance) of the associated network nodes. 

Of particular interest is the ``inverse'' problem first considered in \cite{S_T}: given a network  with adjacency matrix $A$, determine a set of RDPM parameters such that the respective dot products reproduce the adjacency matrix of the observed network with high probability.   In \cite{S_T} Scheinerman and Tucker give an iterative algorithm for determining a set of $n$ separate $1\times d$ vectors $\{X_i\}$, one for each node, so that $(X^TX)_{j,\ell}=\langle x_j, x_\ell\rangle\approx A_{j,\ell}$ for $j\neq \ell$. This procedure also produces a dimension reduction technique by choosing $d<<n$ and representing the nodes by the $\{X_j\}$. As in the generative version, the advantage of this approach for dimension reduction is that $\{X_j\}$ have a natural interpretation in many contexts (discussed above). In Section 4 we will also discuss a null model, based on this dimension reduction technique. 

 A variant of the Schneirman and Tucker algorithm is used in \cite{M0} to learn vectors associated to  an international alliance network. The same approach is used in \cite{T0} to construct an efficient algorithm for inexact graph matching. While the method of \cite{S_T} determines the $\{X_i\}$ by approximating the entries of a given adjacency matrix with a positive semi--definite matrix, the problem of estimating block assignments for SBM derived networks with a RDPM is discussed with a MLE formulation in \cite{O}. Their statistically motivated algorithm offers an asymptotically exact estimation procedure for determining the $\{X_i\}$ corresponding to networks originally drawn from an SBM using a modified RDPM where the connection strengths are a logistic function of the dot products of the associated vectors.

\subsection{Related Work} Our work directly extends the original RDPM formulation \cite{K2,N3}. In that, the edge probabilities are determined by a continuous, $[0,1]$-valued function of the dot product of the associated vectors. That is, step (RDPN 5) computes the probability of edge existence as $\operatorname{Bernoulli}(f(\langle X_j, X_\ell\rangle))$, where $f$ is some continuous function from $\mathbb{R}\rightarrow[0,1]$.  Most applications of this model use the identity function, $f(x)=x$,  \cite{S_T,S_Y2,S_Y1}, although some recent work considers  a logistic function, $f(x)=\frac{1}{1+e^{-x}}$ for more effective community detection \cite{O}. Our generalization differs from \cite{O} in that we replace the Bernoulli distribution with other probability distributions, instead of modifying the function of the dot product. 

There is also a body of work in the creation of generative models for weighted networks. Simple instances of this include the {\em Gaussian ensemble}, wherein nodes are represented by feature vectors of a fixed length with entries drawn from a normal distribution and edges are weighted according to the pairwise correlations \cite{L1}. Null models in the spirit of the configuration model for simple networks can be constructed for particular weighted networks by permuting the edge weights. For more subtle models, the important case of integer weights has been given some attention:  both \cite{R} and \cite{S} give  generative models that use the Poisson distribution to construct an integer weighted network. These can also be used for multinetworks, viewed as integer weighted networks (see e.g. Section 4.1 of \cite{K}). Generative algorithms for {\em multilayer} networks (``stacks" of simple networks that produce a multinetwork) are surveyed in Section 4.3 of \cite{K}.  One common approach is to generate the individual layers independently \cite{DD,SL,K} in which case the distribution of the individual edge weights in the aggregate are sums of independent random variables. The models of \cite{R} and \cite{S} are distinguished from the multilayer approach, as they do not arise as finite sum of independent Bernoulli layers.  In a similar fashion, the Poisson version of our model does not arise as a aggregate of (Bernoulli) RDPM networks. We will show in the next section that the Poisson model presented in \cite{R} is a one--dimensional, restricted version of our approach. Our approach differs from \cite{S} as we do not fix the number of edges that will occur in the network.

Another generative process for weighted networks comes from a weighted version of the SBM. The case of   edge weights  drawn from a Poisson distribution, instead of as binary variables, was used in \cite{B0,K0} to simplify the derivation and construction of the {\em degree-corrected SBM} for simple networks. The degree--corrected SBM adds an additional parameter to each node, reflecting the propensity of that node to form ties. This addresses the problem that many complex networks of interest tend to have hubs and inference based on the standard SBM tends to cluster the nodes by degree, placing all of the hubs together even if they represent separate communities \cite{Z1}. Replacing binary edge weights with Poisson edge weights means that for large values of $n$ and small values of $p$, the Poisson distribution approximates the binomial distribution but is more analytically tractable \cite{GS}. These Poisson versions were applied  to empirical networks \cite{P0,Z1} and found to successfully represent the real world data, as complex networks tend to be large (high $n$) and sparse (low $p$). However, these Poisson SBM approaches are designed to describe unweighted networks not to actually model networks with Poisson valued edge weights.   Aicher et al. have introduced a weighted version of the stochastic block model, using weights drawn from any exponential distribution, for the purposes of community detection \cite{C_1,C}. 

\subsection{Contributions} This paper introduces the Weighted RDPM (WRDPM),  a generalization of the RDPM for weighted networks. Our model shares some characteristics with Aicher et al.'s weighted SBM approach \cite{C_1,C}. However, the WRDPM is more general in that it does not assume an underlying block structure. Additionally, as a latent space model, our model can be studied using linear algebraic and geometric tools.  From a generative perspective, any weighted SBM with positive definite parameter matrices can be realized as a special case of our framework.  The connections between the RDPM with Poisson weights and the stochastic block model have not appeared in the literature previously. 
After describing the formal generative process, we show how  several other generative models arise as special cases of this model. This allows us to define natural generalizations of these models to weighted networks that have not previously appeared in the literature.

Our model provides a principled framework for constructing adjacency embeddings of weighted networks as has been used in \cite{C_M3,S_T,C_M1,C_M2} for simple networks. This embedding is our reason to prefer a latent space model that uses the dot product to parametrize edge weight, as it relates the embedding to a matrix factorization problem as described in Section 4.1. This process allows us to construct geometric interpretations of community structure and node centrality.  We also present the first principled approach to dimension selection for a dot product model by exhibiting a stress function for dimension selection that prioritizes community detection as described in Section 4.4. 

\subsection{Outline} The rest of this paper is organized as follows: Section 2 contains the formal definition of our model and presents an example of the generative process. Section 3 discusses some natural simplifications of our method and relations to other models. The inferential aspect of this model is discussed theoretically in Section 4 and demonstrated with applications in Section 5. 

\section{Weighted Random Dot Product Model}\label{sec:WRDPM}

\subsection{Generative Process} We build on  the RDPM  to produce a  generative model for weighted networks, where the edge weights are drawn from some parametrized probability distribution $P$. The choice of $P$  necessarily depends upon the application under consideration. For non-negative integer data, a Poisson or negative binomial distribution may be most reasonable, whereas continuous data such as correlation coefficients of time series may require a Gaussian or uniform distribution. Our goal is to present a model that is flexible enough to represent all of these varied situations. 

The RDPM constructs a network whose edges are selected according to a Bernoulli distribution parametrized by a function of the dot product of the respective vectors associated to node. This corresponds to a distribution of integer edge weights restricted to the set $\{0,1\}$.  In order extend this framework to arbitrary weighted networks we replace the Bernoulli distribution -- whose mass is concentrated on  $\{0,1\}$ --
with any parametrized distribution $P$ concentrated on the nonnegative reals. Following the RDPM methodology, we associate to each node, a {\em family} of vectors, one for each parameter of $P$. 
Then, we draw the edge weight between each pair of nodes from the distribution parametrized by the collection of respective dot products.
The Weighted Random Dot Product Model (WRDPM) model is thus defined as follows\footnote{We have recently learned that this is effectively the same definition as that proposed in unpublished work of R. Tang \cite{PriebeCommunication}.}:

\begin{enumerate} 
\item[]{\bf(WRDPM 0):} Select a parametrized probability distribution\\ $P(p_1,p_2,\ldots,p_k)$ for the edge weights.  Let  $S_i\subseteq\mathbb{R}$ the domain for $p_i$. 
\item[]{\bf(WRDPM 1):}  Select the number of desired nodes $n$. 
\item[]{\bf(WRDPM 2):}  For each parameter $p_i$,   select a dimension $d_i$. 
\item[]{\bf(WRDPM 3):} For each parameter $p_i$, select a distribution $W_i$ defined over $\mathbb{R}^{d_i}$ so that $\mathbb{P}(\langle X_i, Y_i\rangle\in S_i)=1$ where $X_i$ and $Y_i$ are drawn independently from $W_i$. 
\item[]{\bf(WRDPM 4):}  For each node, $1\leq j\leq n$, select $k$ vectors $1\leq i\leq k$ (one from each parameter space), $X_i^j\in\mathbb{R}^{d_i}$, according to distribution $W_i$. 
\item[]{\bf(WRDPM 5):}  Finally, construct a weighted adjacency matrix, $A$,  for the network, with $A_{j,\ell}$ drawn according to  $P(\langle X^\ell_1,X^j_1\rangle,\langle X^\ell_2,X^j_2\rangle,\ldots,\langle X^\ell_k,X^j_k\rangle)$ for $j>\ell$, $A_{j,\ell}=A_{\ell,j}$  for $j>\ell$ and $A_{j,j}=0$ for all $1\leq j\leq n$.
\end{enumerate}

This process gives rise to an undirected weighted network with no self--loops. The not necessarily symmetric weighted case  can be addressed through an analogous generalization of the  the directed RDPM networks presented in \cite{S_Y1}. We will mostly focus on the case where $P$ is   a distribution over the natural numbers, usually the Poisson distribution.

\subsection{Assortative WRDPM Examples} 

In order to demonstrate the WRDPM process and motivate our later discussion of community structure in WRDPM networks we present two examples that generate Poisson--weighted networks with assortative community structure. Although the setup for these examples is slightly  complicated it is representative of the edge parametrized models discussed in Section 3, the community WRDPM models  of Section 4, and the applications explored in Section 5. 

\subsubsection{Simple Communities}

  \begin{figure}
  \subfloat[Community 1 Vectors]{\includegraphics[height=1.5in]{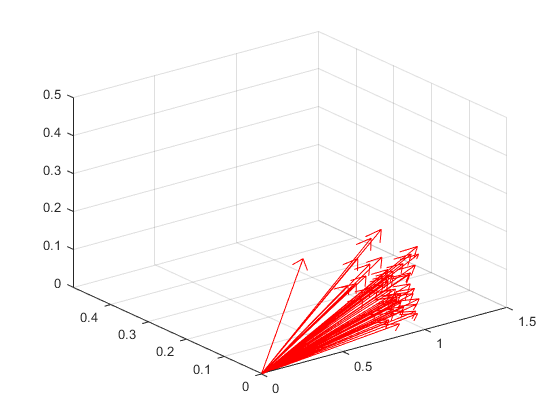}}\ 
    \subfloat[Community 2 Vectors]{\includegraphics[height=1.5in]{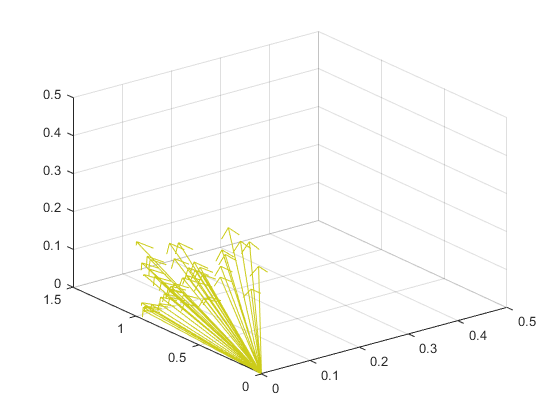}}
      \subfloat[Community 3 Vectors]{\includegraphics[height=1.5in]{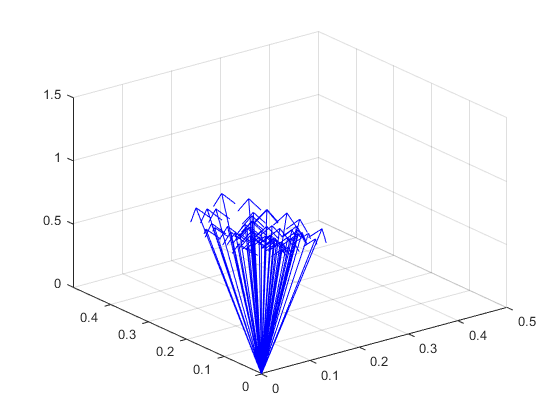}}\\
        \subfloat[All Vectors]{\includegraphics[height=2in]{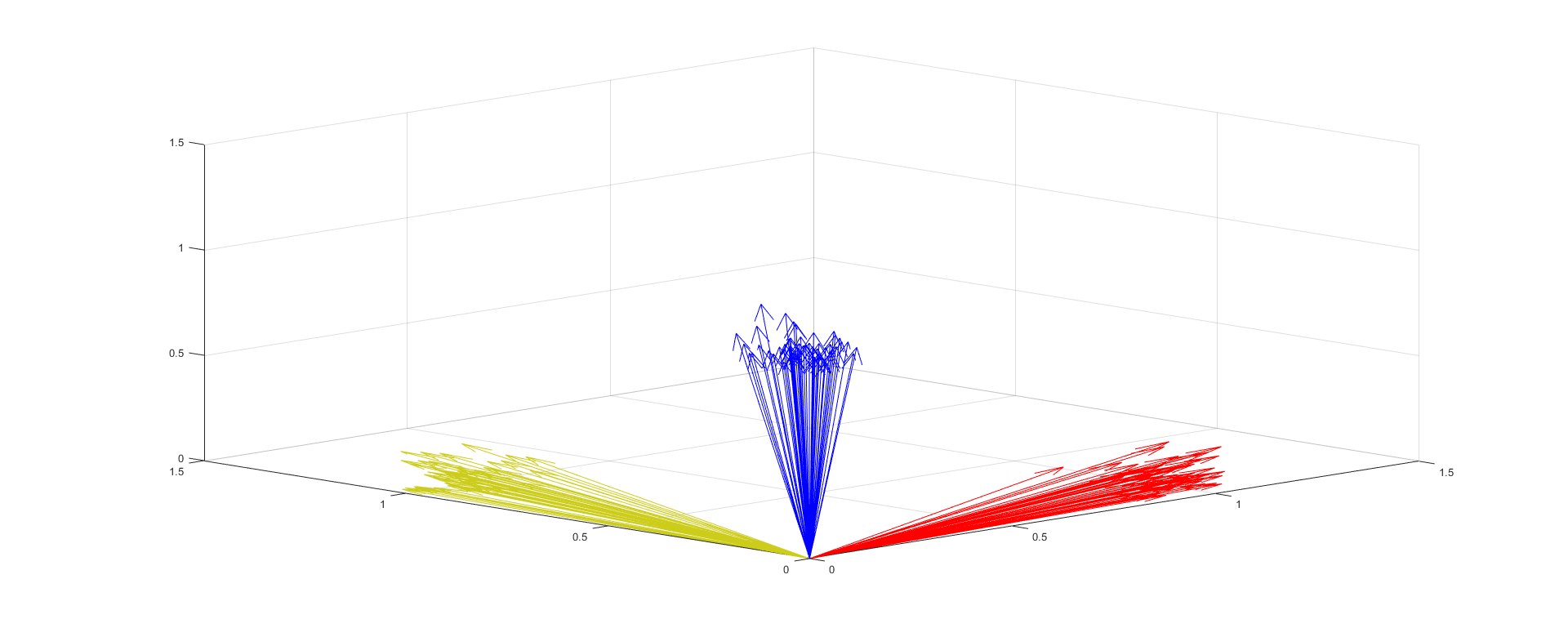}}
\caption{Vectors drawn for a three--community, assortative WRDPM. Subfigures (a), (b), and (c) contain the vectors associated to the individual communities, while (d) shows the entire collection of vectors. The assortative structure of the resulting network (Figure 2) can be determined from (d) as the intra--community dot products are much larger than the inter--community dot products.}
  \end{figure}
  
  \begin{figure}
   \subfloat[Dot Products]{\includegraphics[height=2in]{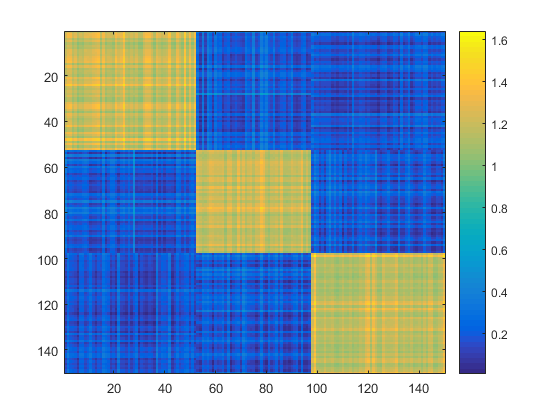}}\qquad 
    \subfloat[WRDPM Network]{\includegraphics[height=2in]{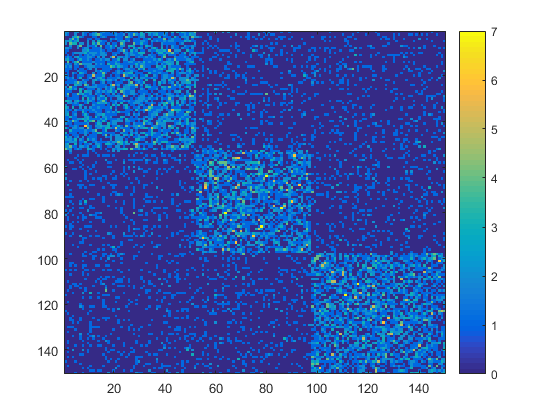}}
    \caption{Once the vectors are drawn from $W_\lambda$ (Figure 1) we compute the pairwise dot products (a) to determine the Poisson parameter for each edge. Then, the weight of each edge is drawn as a Poisson random variable using the dot product parameter forming an assortative weighted network (d). }
  \end{figure}
 
The adjacency matrix of a network with assortative community structure can be approximated as a block matrix where the diagonal blocks are denser than the off--diagonal blocks. In order to encode this property with the WRDPM  we want the matrix of dot products in {\bf (WRDPM 5)}to have this block structure.  A natural way to approach this it to begin by selecting a family of $k$ orthogonal vectors, one for each desired community. Then we form a distribution for each community by  add a small amount of variation to each vector in the directions of the other communities vectors.  The distribution $W$ is then taken to be a  distribution over these families of ``nearly orthogonal'' vectors, for example by choosing the community assignment for each node uniformly over the $k$ vector families. 

As an explicit three--community example we construct a WRDPM using this approach. We begin by selecting $P$ to be the Poisson distribution for the edge weights {\bf (WRDPM 0)} and take $n=150$ {\bf (WRDPM 1)}. Since the Poisson distribution has only one parameter, $\lambda\in\mathbb{R}^+$ we only need to select one dimension,  $d_\lambda$. As this is a three--community example we take $d_\lambda=3$ {\bf (WRDPM 2)}. Motivated by the discussion in the preceding paragraph we define the following distribution over $\mathbb{R}^3$ {\bf (WRDPM 3)}.
Let $\{e_1,e_2,e_3\}$ be the standard basis of $\mathbb{R}^3$.  and $y$ be a random variable with a half--normal distribution centered at $0$ with variance $.1$. Then, we define $W_\lambda$ by uniformly  selecting one of the basis vectors and adding noise with draws from $Y$:

\begin{equation}
W_\lambda=\begin{cases} 
e_1+Ye_1+Ye_2+Ye_3&\frac13\\
e_2+Ye_1+Ye_2+Ye_3&\frac13\\
e_3+Ye_1+Ye_2+Ye_3&\frac13\\
\end{cases}
\end{equation}

  For each node we select a vector from $W_\lambda$ {\bf (WRDPM 4)}. Figure 1(a-c) shows the vectors drawn in each of the individual communities, while Figure 1(d) shows all of the vectors together. In Section 5 we will see that vectors learned from U.S. Senate voting data exhibit similar behavior to the vectors drawn from $W_\lambda$ as shown in Figure 1(d). To complete the WRDPM we compute the pairwise dot products of the vectors drawn in the previous step {\bf (WRDPM 5)}. The dot products are displayed in Figure 2(a). The displayed values are the parameters used for actually constructing networks in this example. The heat map of a sample network parametrized by these values {\bf (WRDPM 5} is shown in Figure 2(b).

 \subsubsection{Multiresolution WRDPM}

In many complex networks the structure within communities is quite different than the structure connecting communities. The multiresolution model presented in \cite{Fos1} addresses this issue by modeling intra--community connections with a latent space model and the inter--community connections with a SBM. We can construct a similar model by modifying the distribution in {\bf (WRDPM 3)} to encode a weight parameter for each vector, representing the propensity of the associated node to form links. The distribution of the weights then controls the intra--community distribution of edges.
 
In this example we take the weight parameter to be exponentially distributed. More formally, let $X$ be an exponential random variable with parameter $\lambda=2$. Then we take our vector distribution for {\bf (WRDPM 3)} as follows:
 
\begin{equation}
V_\lambda=\begin{cases} 
Xe_1+Ye_2+Ye_3&\frac13\\
Xe_2+Ye_1+Ye_3&\frac13\\
Xe_3+Ye_1+Ye_2&\frac13\\
\end{cases}
\end{equation}
 
This version of the WRDPM also generalizes the model presented in \cite{R}, as the connection behavior within each community is mainly governed by a single multiplicative parameter, the exponential weight given by draws of $X$. In order to highlight this behavior the nodes in each community are sorted by their intra--community weight for the plots in Figure 4.  

To construct the WRDPM we proceed as in the previous example, selecting $P$ to be Poisson {\bf (WRDPM 0)}, $n=150$ {\bf (WRDPM 1)}, and $d_\lambda=3$ {\bf (WRDPM 2)}. As in the previous example we display the individual community vectors Figure 3(a--c) as well as the entire collection of vectors Figure 3(d). The dot product and weighted network plots in Figure 4 display the logarithms of the values to account for the exponential scaling. From these plots, the intra--community structure is clear unlike the random, noisy behavior observed in Figure 2(d). This type of structure will appear again when we analyze the coauthorship network in Section 5. 

   \begin{figure}
  \subfloat[Community 1 Vectors]{\includegraphics[height=1.5in]{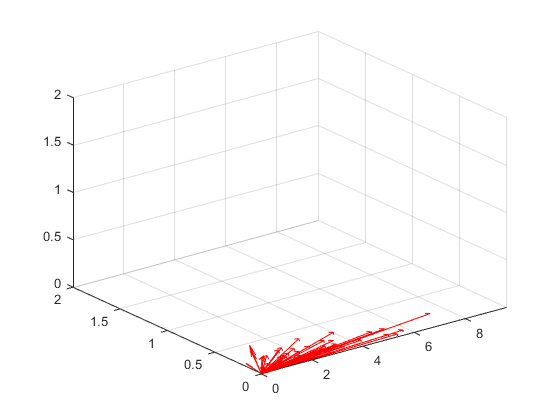}}\ 
    \subfloat[Community 2 Vectors]{\includegraphics[height=1.5in]{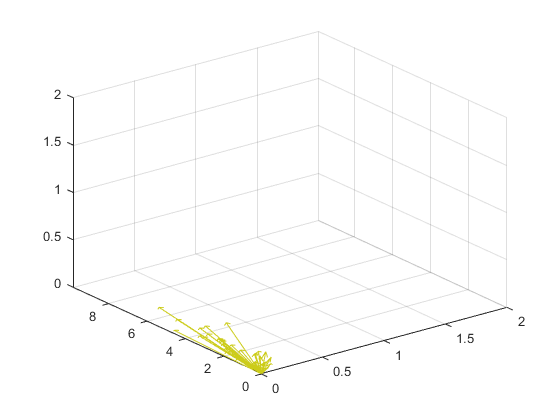}}
      \subfloat[Community 3 Vectors]{\includegraphics[height=1.5in]{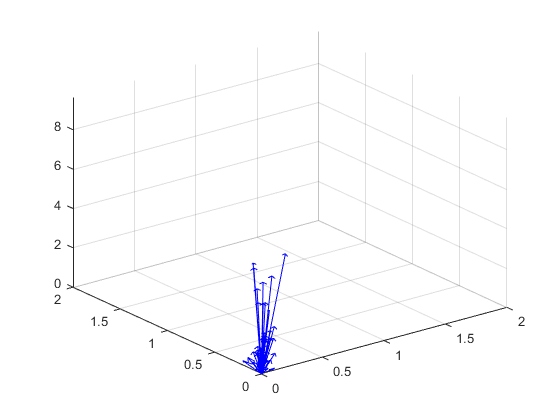}}\\
      
      \begin{center}
        \subfloat[All Vectors]{\includegraphics[height=3in]{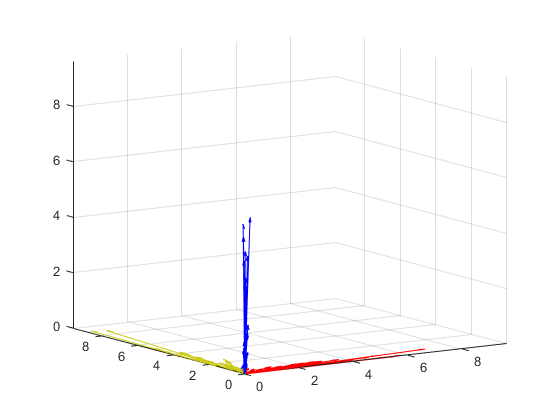}} \end{center}
\caption{Vectors drawn for a three--community, assortative WRDPM with multiresolution structure. Subfigures (a), (b), and (c) contain the vectors associated to the individual communities, while (d) shows the entire collection of vectors. The assortative structure of the resulting network (Figure 2) can be determined from (d) as the intra--community dot products are much larger than the inter--community dot products.}
  \end{figure}
  
  \begin{figure}
   \subfloat[Dot Products]{\includegraphics[height=2in]{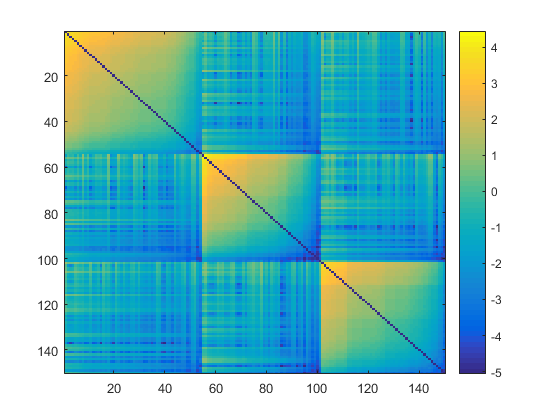}}\qquad 
    \subfloat[WRDPM Network]{\includegraphics[height=2in]{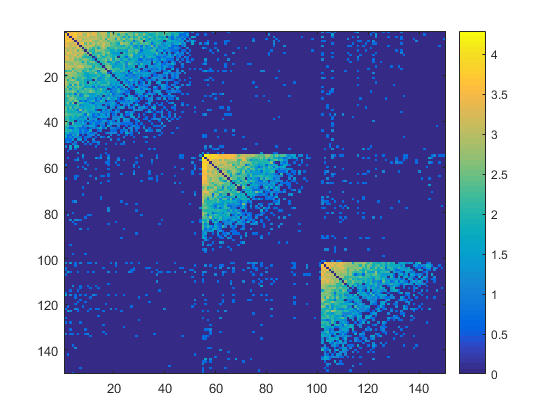}}
    \caption{Once the vectors are drawn from $V_\lambda$ (Figure 3) we compute the pairwise dot products (a) to determine the Poisson parameter for each edge. Then, the weight of each edge is drawn as a Poisson random variable using the dot product parameter forming an assortative weighted network (d). In this figure the logarithm of the dot products and edge weights are displayed to account for the scaling of the exponential weights.}
  \end{figure}
  
\section{Specializations}

The original RDPM, as well as the WRDPM as defined in the Section 2.1, describe a very broad class of models due to the arbitrary choice of vector distributions. Many previously studied generative models for (weighted) networks can be realized as special cases of the (W)RDPM by restricting the dimension of the latent spaces or the distributions of vectors. In this section, we focus on network models with prescribed edge connection parameters. The advantage to realizing these models as special cases of the WRDPM is that we obtain a latent space interpretation, and hence additional structure, in order to study them. The RDPM--motivated theoretical results about the SBM make use of this approach \cite{O,C_M1}.

\subsection{Edge Parametrized Models}

Many generative models for networks can be described by a set of $\binom{n}{2}$ Bernoulli parameters, one for each possible edge. A network is constructed from the model by drawing the edges independently based on the given parameters. The Erdos--Renyi, SBM, and Chung--Lu models all fall in to this category. For example, in the Erdos--Renyi model all $\binom{n}{2}$ parameters are the same while in the SBM the parameters depend only on the community assignment of the nodes. Although our main interest is in these commonly studied and applied models, we begin in a more general setting, analyzing arbitrary prescribed connection parameters. 

In general, our approach mirrors the examples in Section 2.2. We identify the relevant structure and select a collection of vectors whose respective dot products realize the given edge connection parameters. This approach highlights a difference that appears in the literature with respect to describing SBMs with RDPMs. Although the original generative definition of the RDPM in \cite{N3,S_Y2} follows the method outlined in Section 1.2, later versions that are focused on applications to SBM problems, such as \cite{O,C_M1}, adopted a modified version where nodes are assigned directly to pre--selected vectors, instead of drawing a vector from a distribution for each node. We will discuss another setting where the second interpretation is  natural in Section 4. 

Both versions are useful in different contexts. They can be related by selecting the distribution $W$ for the original version to be the uniform distribution over the pre--selected vectors from the second approach as in the examples of Section 2.2. This does not give rise to equal distributions over $n$ node networks but it does match the interpretations in expectation. 

Given a set of $\binom{n}{2}$ edge parameters $\{a_{j,\ell}\}_{1\leq j<\ell\leq n}$, for a generative model, it is convenient to form an $n\times n$ matrix $A$ with undetermined diagonal by setting $A_{j,\ell}=A_{\ell,j}=a_{j,\ell}$. In order to realize the  model with the WRDPM we need to find $1\times d$ vectors $\{X_j\}_{1\leq j\leq n}$ so that $\langle X_j, X_\ell\rangle=A_{j,\ell}$. Letting $X$ be the matrix whose rows are given by the $\{X_j\}$ this would imply $(XX^T)_{j, \ell}=A_{j, \ell}$ for $j\neq \ell$. Thus, in order for the $\{X_j\}$ to exist there must be a choice of diagonal entries for $A$ such that $A$ is positive definite. The following result guarantees that we can always find such a representative. 

\begin{theorem}
Let $n$ be a fixed positive integer. For each pair $(i,j)$ with $1\leq i<j\leq n$ let $a_{i,j}=a_{j,i}\in \mathbb{R}$. Then there exist $n$ real numbers $a_{\ell,\ell}$ for $1\leq \ell\leq n$ such that the matrix $A_{i,j}=a_{i,j}$ is positive definite. 
\end{theorem}
\begin{proof}
Let the $a_{i,j}$ be selected arbitrarily. For $1\leq \ell\leq n$ choose $a_{\ell,\ell}\in\mathbb{R}$ so that $a_{\ell,\ell}> \sum_{j\neq \ell}|a_{j,\ell}|$. Form a matrix $A$ with $A_{i,j}=a_{i,j}$. This is a real symmetric matrix and so by the spectral theorem $A$ has real eigenvalues. 

Applying Gershgorin's Circle Theorem to $A$ gives that the eigenvalues of $A$ lie in the closed disks centered at $a_{\ell,\ell}$ with radius $\sum_{j\neq \ell}|a_{j,\ell}|$. Intersecting these disks with the real line gives that the eigenvalues of $A$ must lie in $\bigcup_{\ell=1}^n \left[a_{\ell,\ell}-\sum_{j\neq \ell}|a_{j,\ell}|, a_{\ell,\ell}+\sum_{j\neq \ell}|a_{j,\ell}|\right]\subseteq \mathbb{R}^+$. Thus, all eigenvalues of $A$ are positive and $A$ is positive definite. 
\end{proof}
This theorem allows us to recover any edge parameterized  graph model as a special case of the WRDPM. 
\begin{corollary}
Any generative network model, on a fixed number of nodes $n$, where the edge weight between each pair of nodes is drawn independently from a fixed probability distribution, possibly with different parameters for each pair, can be realized under the WRDPN.
\end{corollary}

\begin{proof}
Consider a fixed instance of any such generative model. Let $P$ be the $k$--parameter distribution from which the edge weights are drawn and for $1\leq i\leq k$ let $a_{j,\ell}^i=a^i_{\ell,j}$ be the value of the $i$th parameter between nodes $j$ and $\ell$. Applying Theorem 1 to the collection $a_{j,\ell}^i=a^i_{\ell,j}$ gives a positive definite matrix $A^i$. Thus, there exists an $n\times n$ matrix $X^i$ such that $(X^i)^TX^i=A$.

To form the WRDPM  that matches the given generative model we take $d_i=n$ for all $1\leq i\leq k$ and to each node $1\leq j\leq n$ assign the collection of vectors given by the $j$th columns of the $X^i$ for $1\leq i\leq k$. Then, this WRDPM  defines the same distribution over weighted graphs as the original generative model. 
\end{proof}

Note that although Corollary 1 determines a version of the WRDPN  that matches a given generative model, there are many ways to choose vectors satisfying the conditions. For example, for any choice of non--diagonal values, Theorem 1 determines an $n$--dimensional half--space of positive definite matrices, each leading to a different assignment of vectors to nodes. As an explicit example, note that the absolute value of the Laplacian $|L|=D+A$ satisfies this condition. Additionally, although the result shows that there is a WRDPM  with $d_i=n$ for all $1\leq i\leq k$ it is often possible to find a lower dimensional choice of $X^i$ that suffices. 

We next examine the implications of Corollary 1 for three specific, well--studied models; the Erdos--Renyi Model, the stochastic block model, and the Chung--Lu model. In particular, for each case we discuss the most efficient choice of $d$ as well as weighted generalizations of these models.

\subsection{Erdos--Renyi Networks}

In the Erdos--Renyi random graph model each edge occurs with some fixed probability $p\in(0,1)$. That is, edges are drawn from the Bernoulli distribution over $\{0,1\}$ with parameter $p$ for all pairs of nodes. In this case, although it is possible via Corollary 1 to obtain an $n$ dimensional embedding, in fact there exists an embedding for any $d\geq 1$ obtained by selecting a single vector in $\mathbb{R}^d$ with squared norm equal to $p$. Choosing $W$ in {\bf (WRDPN  3)} to be the constant distribution on that vector gives a WRDPN  equivalent to the Erdos--Renyi model.

To generalize this model to weighted networks we replace the Bernoulli distribution with another parameterized probability distribution, $P$, and enforce that the edge weight between each pair of nodes be drawn from a single parametrization of $P$. As long as the parameters are assumed positive, we may select the $d_i$ arbitrarily and for each $1\leq i\leq k$ take $W_i$ to be the constant distribution over a single vector $X_i$. Then, each edge weight is drawn from the distribution $P(\langle X_1,X_1\rangle,\ldots,\langle X_k,X_k\rangle)=P(||X_1||^2,\ldots,||X_k||^2)$. 

A standard use of the un--weighted Erdos--Renyi model is as a null model, where $p$ is chosen as the number of edges divided by $\binom{n}{2}$. For the Poisson examples from Section 2.2 we can use a similar process, using the MLE to estimate $\lambda$ as the sum of the edge weights divided by $\binom{150}{2}$. We use this process to create Poisson Erdos--Renyi networks based on the weighted networks in Figure 2(b) and Figure 4(b). Examples of these (unstructured) networks are displayed in Figure 5.  

Comparing the weighted clustering coefficients \cite{B4} of these null models to networks drawn from the WRDPM highlights the differences between the simple community and multiresolution models. For each example, we estimated $\lambda$ from the WRDPM network and computed the average weighted clustering coefficient for several draws from the Erdos--Renyi null model shown in Figure 6. The simple communities model Figure 6(a) displays higher clustering than the Erdos--Renyi version while the multiresolution model Figure 6(b) displays lower clustering. Both of these results are expected. The simple communities model is designed to have the majority of its edges in fairly homogenous, dense communities. The intra--community structure in the multi--resolution blocks has a hierarchical structure with much less clustering between the lower weight nodes.

\begin{figure}
\subfloat[Assortative Null Model cf. 2(b)]{\includegraphics[height=2in]{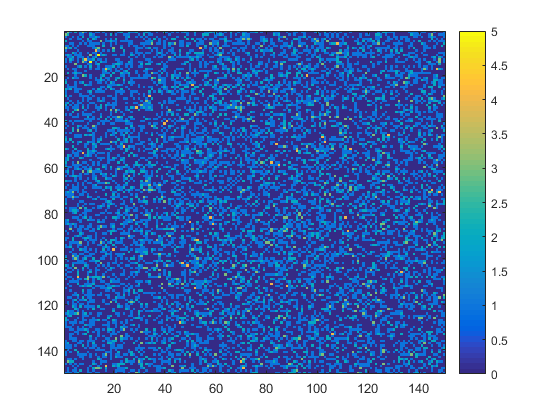}}\qquad
\subfloat[Multiresolution Null Model cf. 4(b)]{\includegraphics[height=2in]{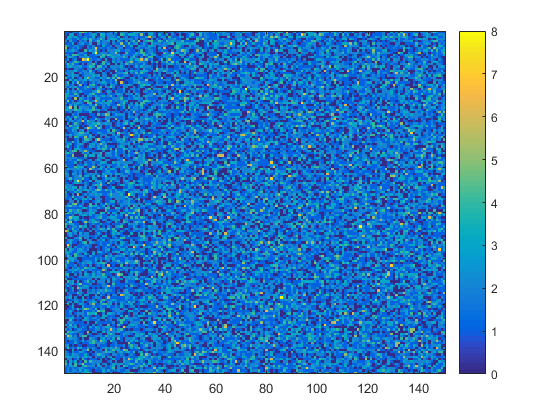}}
\caption{Erdos--Renyi null model comparison versions of the examples in Section 2.2. The estimated $\lambda$ from Figure 2(b) is $.5233$ and the estimated $\lambda$ for Figure 4(b) is $1.5532$. As expected these networks have much less structure than the examples   in the previous section. }
\end{figure}

\begin{figure}

\subfloat[Assortative Null Model cf. 2(b)]{\includegraphics[height=2.5in]{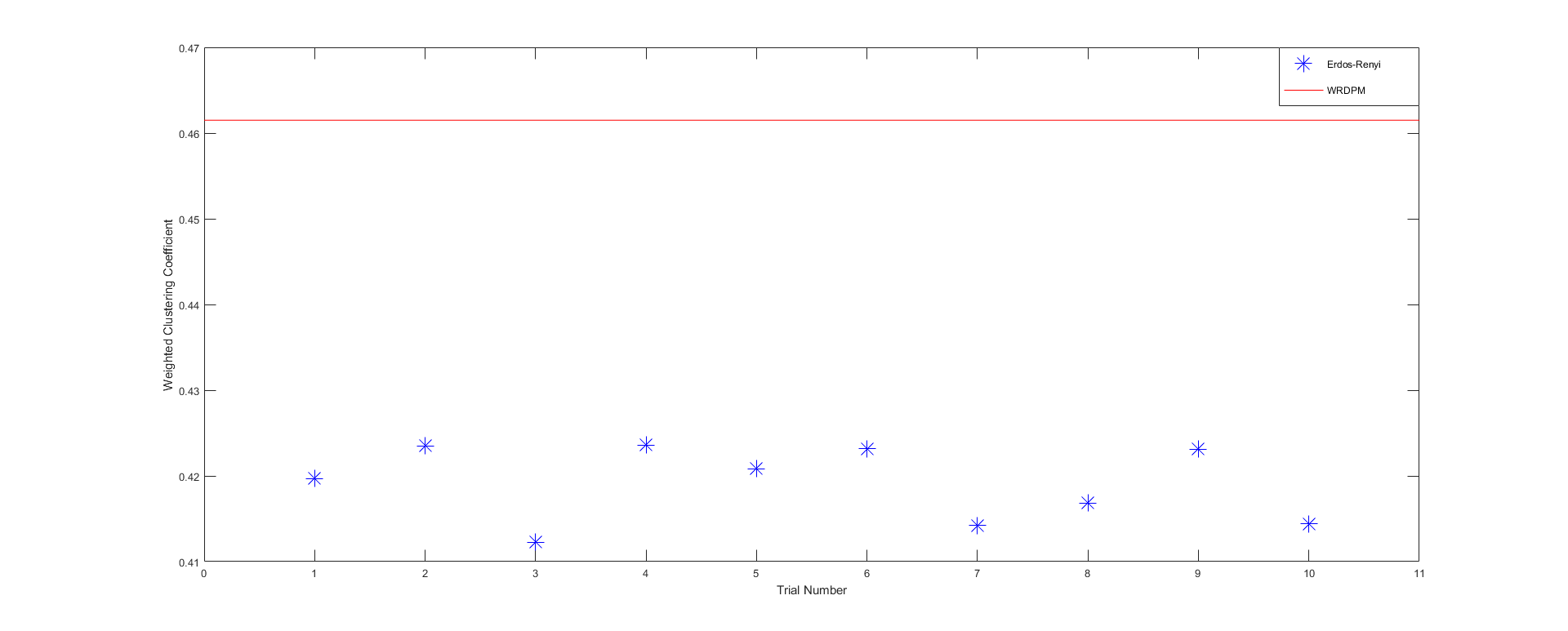}}\\
\subfloat[Multiresolution Null Model cf. 4(b)]{\includegraphics[height=2.5in]{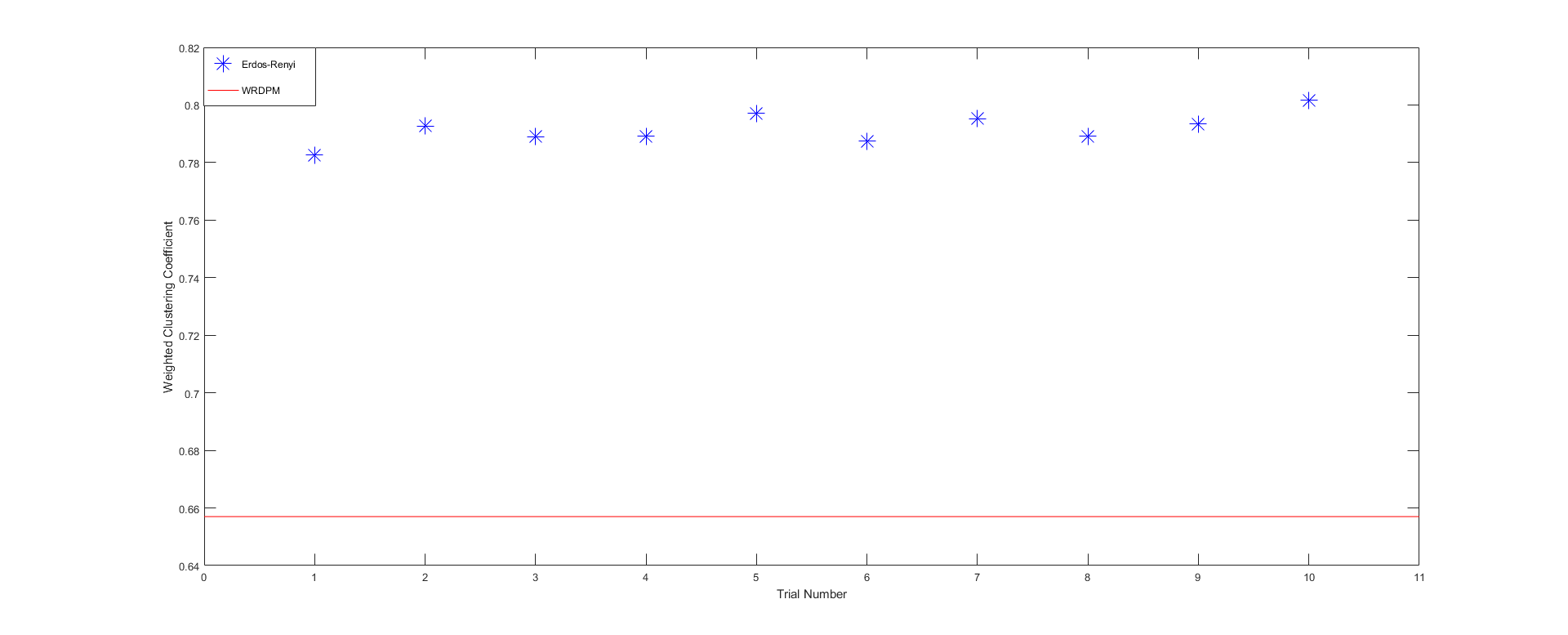}}
\caption{Comparison of average weighted clustering coefficient between the WRDPM examples of Section 2.2 and the derived Poisson Erdos--Renyi null models. For the simple community structure (a) the WRDPM network exhibits higher clustering than the null model, while the multiresolution WRDPM exhibits lower clustering. Both of these results are expected based on the corresponding un--weighted versions.  }
\end{figure}
\subsection{Stochastic Block Models}

 In order to describe a block model in the WRDPM, where the edge parameters between any pair of nodes depends only on the community assignment of those nodes, we restrict the $W_i$ to finite distributions. That is, for any choice of edge weight distribution $P$ with $k$ parameters, selecting $W_i$ to be a distribution with finite support for all $1\leq i\leq k$ gives rise to a weighted block model where the community assignment is defined by collections of nodes who are assigned to the same vector in each parameter space. When $P$ is the Poisson distribution, this recovers the model used in \cite{B0,K0} as a replacement for the standard Bernoulli model. As noted at the beginning of this section we could alternatively follow \cite{O,C_M2}, making the community assignments for the nodes in advance and then assigning nodes to vectors based on the previously selected communities. 

In the traditional stochastic block model, each node in the network is assigned to one of $b$ communities and the edges are drawn as Bernoulli random variables with probabilities determined by the community assignments. These parameters are usually summarized in a $b\times b$ matrix $B$ where the Bernoulli parameter between community $x$ and community $y$ is given by $B_{x,y}=B_{y,x}$. When the matrix $B$ is positive (semi--)definite it is possible to apply the method used to prove of Corollary 1 to $B$ to obtain a collection of $b$ vectors, one for each community, so that assigning each node to its respective community's vector realizes the block model as a WRDPM.

As an example, consider the SBM defined by community matrix:
$$B=\begin{array}{rcl}
.5&.05&.1\\.05&.4&.05\\
.1&.05&.3
\end{array}
$$

We can factor $B$ as:
\begin{align*}
B=&\begin{bmatrix}
-.3326 &-.4164 &-.8462\\
-2.004&.9079&-.3681\\
.9215&.0471&-.3854
\end{bmatrix}\begin{bmatrix}
.2530&0&0\\0&.3797&0\\0&0&.5673
\end{bmatrix}\begin{bmatrix}
-.3326 &-.4164 &-.8462\\
-2.004&.9079&-.3681\\
.9215&.0471&-.3854
\end{bmatrix}^T\\
=&\begin{bmatrix}
-.1673&-.2566&-.6373\\
-1.008&.5594&-.2772\\
.4636&.0290&-.2903
\end{bmatrix}\begin{bmatrix}
-.1673&-.2566&-.6373\\
-1.008&.5594&-.2772\\
.4636&.0290&-.2903
\end{bmatrix}^T\\
=&XX^T
\end{align*}

Hence, if we take $W$ in  {\textbf{(WRDPM 3)} to be the uniform distribution over $x_1=[-.1673, -.2566, -.6373], x_2=[-1.008, .5594, -.2772]$, and $ x_3=[.4636, .0290, -.2903]$ we recover the original SBM defined by $B$ since $\langle x_j, x_\ell\rangle=B_{j,\ell}$ by construction. Again, following \cite{O,C_M1} we can instead determine the community assignments beforehand and assign  the nodes in the first community to $x_1$, the nodes in the second community to $x_2$ and the nodes in the third community to $x_3$.

When $B$ is not positive (semi--)definite it is not possible to give a general bound, 
 independent of $n$, on the smallest possible value of $d$, i.e. the minimum number of vectors necessary to represent the SBM by the WRDPM, as is shown by the two community model with $B_{1,1}=B_{2,2}=0$ and $B_{1,2}=B_{2,1}=1$, since it requires at least n/2 dimensions to realize as a WRDPM. However, these block models can be realized under the WRDPM  using $d=n$ and can frequently be analyzed using fewer dimensions. Some results discussing when a SBM can be realized with $d=b$ in versions of the RDPM can be found in \cite{O}.

The non--positive definite example above also cannot be represented by two vectors, as the vectors for any two nodes in the same community must be orthogonal. As we will discuss in Section 4.4, the magnitude of the vector assigned to each node have an interpretation in terms of a betweeness centrality.  Thus, for modeling SBMs with the WRDPM it is natural to enforce that nodes in the same community have vectors of the same magnitude,  i.e. each node in a given community has the same a priori ability to transfer information across the network. Note that the embedding obtained by factoring the matrix with the diagonal entries $ n \sum_{z=1}^b B_{j,z} $ is a natural choice with this property.

The weighted SBM introduced in \cite{C_1,C} can be modeled by the WRDPM in a similar fashion. In this setting, the edge weights are drawn from a parametrized exponential distribution, with the parameters again only depending on the community assignments of the nodes. For each parameter $i$, we obtain a separate matrix of intra-- and inter-- community values which can then be represented by a choice of dimension $d_i$ and set of vectors $W_i$.  As discussed in Section 1, one of the main applications of the RDPM is providing a principled, theoretical framework for proving results about the SBM \cite{O, C_M2}. The WRDPM provides this same structure for the WSBM by giving a natural, geometric interpretation to the parameterization of the community connections.

\subsection{Chung--Lu Networks}

For a fixed number of nodes $n$ the Chung--Lu model is parametrized by a collection of weights, $\{w_\ell\}$, one for each node, that encode the expected degree sequence \cite{FC1}. The probability of placing an edge between node $j$ and node $\ell$ is given by $\frac{w_j\cdot w_\ell}{\sum_{a=1}^nw_a}$. This is a generalized version of the configuration model that is significantly more tractable for proving general results on expected network metrics. To realize a Chung--Lu model in the WRDPN  we can select $d$ arbitrarily and then select a vector $X_0\in\mathbb{R}^d$ with $||X_0||^2=\frac{1}{\sum_{a=1}^nw_a}$. To each node $1\leq j\leq n$, we associate the vector $w_jX_0$. Then, the dot product of the vectors associated to any two nodes $1\leq j\neq \ell\leq n$ is given by $\langle w_jX_0,w_\ell X_0\rangle=\frac{w_j\cdot w_\ell}{\sum_{a=1}^nw_a}$.

This model can be generalized by selecting an arbitrary dimension $d$, and a single vector $X_0\in\mathbb{R}^d$. Then $W$ can be chosen as a distribution over $\mathbb{R}^+X_0$. That is, $W$ is a distribution over the ray from the origin through $X_0$. Alternatively, $W$  can thought of as a distribution of weights or strengths associated to the nodes. When $P$ is the Poisson distribution, this is equivalent to the generative model introduced in \cite{R}, where each node is associated to a positive real number. 

Generalizing this approach to other distributions we can select a single vector in each parameter space and take the $W_i$ to again be distributions over the ray through each vector. As in the original Chung--Lu model, the vector assigned to each node can be thought of as representing the expected strength or weight of each node in the weighted network. For example, in the Poisson case, parametrized by the mean, the expected weight of an edge between nodes $j$ and $\ell$, associated to vectors $w_jX_0$ and $w_\ell X_0$, is exactly $w_jw_\ell||X_0||^2$. 

\section{Community Detection}
\subsection{Methods} 

In this section, we consider the inverse problem for the WRDPN  for Poisson weighted networks. That is, given a specific weighted network of interest, we attempt to find a collection of $d$--dimensional vectors $\{X_i\}_{i=1}^n$ that best represent the weighted network as a WRDPN . Letting $A$ be the adjacency matrix of our weighted network, we want to choose the $X_i$ so that $\langle X_i, X_j\rangle \approx A_{i,j}$ for $i\neq j$. Defining $X$ to be the matrix whose columns are the $X_i$, this is equivalent to approximating the non--diagonal entries of $A$ with the non--diagonal entries of $X^tX$. In other words, we are trying to solve a restricted matrix factorization and dimension reduction problem by minimizing the Frobenius norm of $X^tX-A$. 

Since the Poisson distribution is parametrized by the mean, and we assume that each $A_{i,j}$ for $i\neq j$ is a Poisson random variable, the iterative matrix factorization algorithm given in \cite{S_T} generalizes naturally to the weighted network setting, allowing us to estimate the vectors $\{X_i\}$, up to an orthogonal transformation. This algorithm is not guaranteed to converge, as there exist pathological examples with poor limiting behavior \cite{S_T}. However, when the algorithm does converge the solution is guaranteed to be a local optimum \cite{S_T}, and for data generated weighted networks the algorithm seems to converge rapidly in practice on both synthetic and empirical networks. 

Once we have such an embedding, we may use geometric and linear algebraic techniques to study the vectors as a proxy for the data--generated weighted networks as in other latent space methods.  In particular, Scheinerman and Tucker also introduced an angular $k$--means algorithm for clustering learned vectors \cite{S_T}. This method also generalizes to the weighted network setting, as the embedding itself reflects the dot product similarity. Thus, the interpretations,  described in \cite{S_T,S_Y1} of the vector directions, representing similarity in link formation patterns between the nodes, and the magnitude of the vectors, representing propensity to communicate, are still present in this model.

\subsection{Communities} The community structure of a weighted network has particularly strong connection to the geometry of  the associated embedding. As the directions of the vectors in the embedding capture a measure of similarity in link--formation patterns between nodes, nodes that belong to the same community tend to point in similar directions. Practically, this means that embeddings of weighted networks with particularly well--defined communities will tend to separate into nearly orthogonal components, with one subspace per community. The examples in Section 2.2 have this property, with one community corresponding to each of the standard basis vectors for $\mathbb{R}^3$. The examples in Section 5 also exhibit this behavior.

To further illustrate this point, consider a  network consisting of $\ell$ disjoint cliques as in Figure 7. If our embedding has $\ell$ dimensions, we can optimize our approximation of $A$ with $X^TX$ by assigning each community to a separate, orthogonal one--dimensional subspace and assigning each vector a length of $1$. In this case, increasing the number of dimensions will not yield a better embedding -- in terms of the Frobenius norm-- as the off diagonal entries of $A$ and $X^TX$ agree exactly. 

If the disjoint communities have more structure, within each subspace, we can assign magnitudes to each vector relative to the weighted degree of each node as in Figure 7 (d). When $d=\ell$ the magnitude of the individual vectors is the only free parameter to adjust to match the communities, giving rise to a graded block model. However, in this case of heterogeneous intra--community structure, given more dimensions, we can more accurately approximate the communities underlying structure by embedding each separately into its orthogonal component.

\begin{figure}[!h]
\centering
\subfloat[Disjoint Cliques]{\includegraphics[height=1.15in]{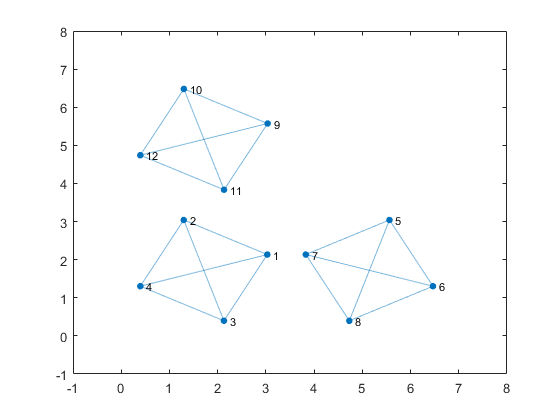}}  \ 
\subfloat[Weighted Clusters]{\includegraphics[height=1.15in]{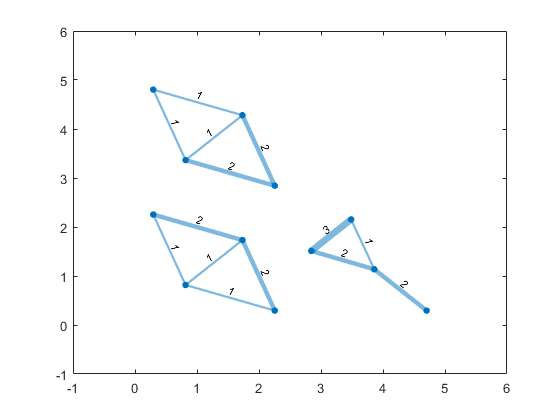}}\\
\subfloat[Embedding of (a)]{\includegraphics[height=1.15in]{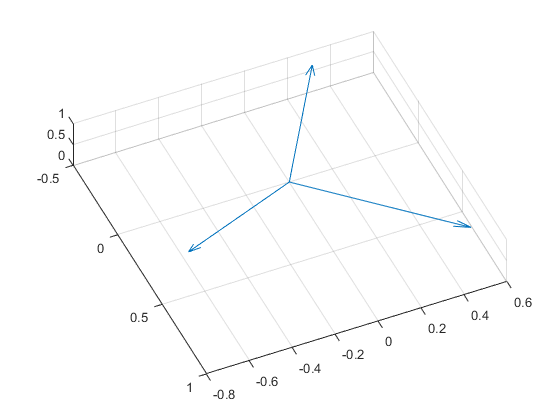}}  \ 
\subfloat[Embedding of (c)]{\includegraphics[height=1.15in]{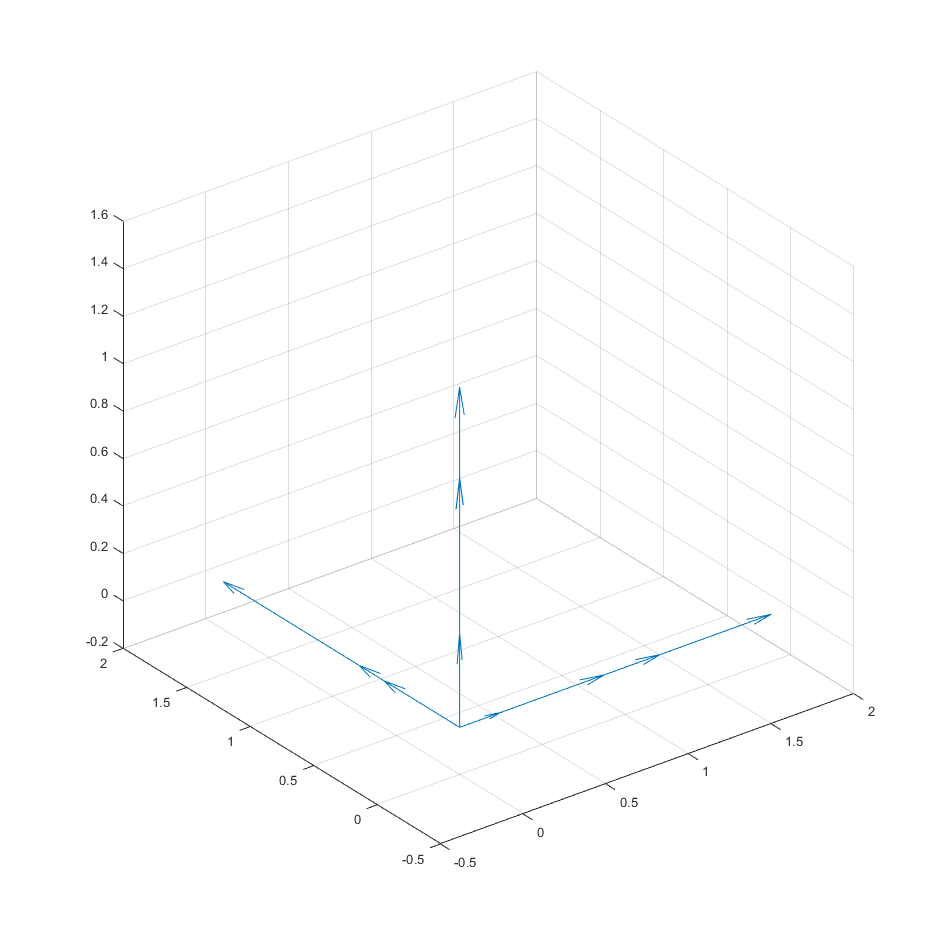}}
\caption{Vector embeddings associated to disjoint communities. For the case of disjoint cliques, there are only three vectors as each node in a clique is assigned to the same vector. For the weighted clusters each node is assigned a magnitude reflecting its connection strength within the cluster. Thus, there are  10 separate vectors in the embedding (d) of (b) as two pairs of nodes are indistinguishable.}
\end{figure}
\subsection{Centrality} While the community structure of a network can be revealed by the angles obtained from the WRDPN  embedding, the lengths of the vectors capture a measure of centrality of their associated nodes. Previous work on the RDPM has associated vector length with propensity to communicate, a measure of degree \cite{S_T,S_Y1}. This is natural, as for any fixed vector $x$, an increase in magnitude will increase its dot product with each other  vector linearly, leading to a higher probability of edge formation overall. However, in light of our discussion of community structure, we can see that vector length is also related to a betweeness measure that increases for nodes that are incident to multiple communities. 

Returning to our toy example of the previous section with $\ell$ disjoint cliques, select one node in each clique and connect it to each of the other selected nodes. In this case, we can obtain an exact embedding into $\ell+1$ dimensions by assigning each cluster to a orthogonal subspace as before and using the last dimension to assign a one to the selected nodes and a zero to all other nodes. Figure 8 shows an example of this process with $\ell=2$. Then, the length of the of the nodes connecting the communities is $\sqrt{2}>1$. Additionally, deleting edges between the selected nodes and their original communities will still lead to an embedding where the connecting nodes, those with high betweeness, are assigned longer vectors, even if they have smaller degree than the regular community members. 

This toy example is a reflection of a general property that if the embedding separates into nearly orthogonal communities then vectors associated to nodes that are adjacent to nodes in other communities must have greater length to compensate for the nearly orthogonal angles of the vectors. This property is a refinement from the usual interpretation of vector length as a proxy for number or weight of connections formed by an individual node and helps to paint a more complete picture of the information captured by the WRDPN  embedding process.

\begin{figure}[!h]
\centering
\subfloat[Toy Network]{\includegraphics[height=1.15in]{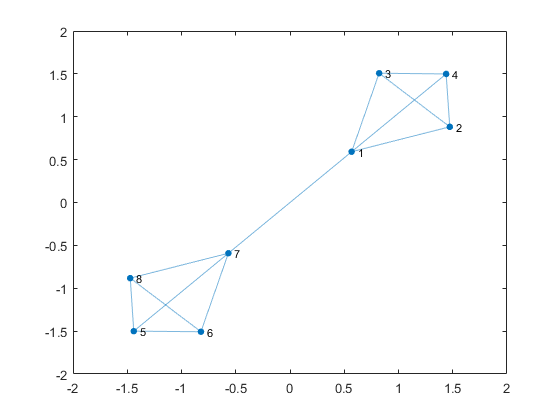}} \
\subfloat[Vector Embedding]{\includegraphics[height=1.15in]{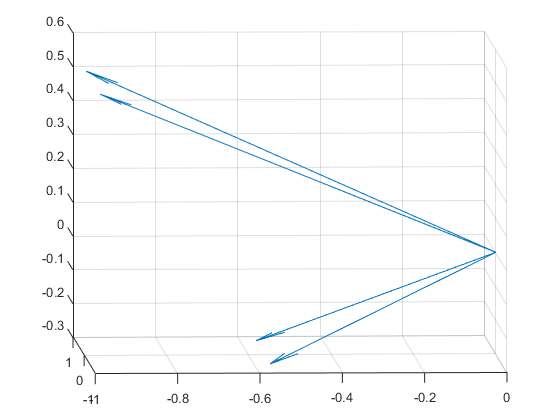}}
\caption{A toy network (a) and its associated WRDPN  embedding (b) demonstrating the effect of inter--community edges on vector magnitude. The nodes whose edges are only witin a community are assigned to the shorter vectors with negative $z$ component, while the vectors connecting the communities are associated to the longer vectors wit positive $z$ component.   }
\end{figure}

\subsection{Dimension Selection} 

In order to determine an appropriate value of $d$ for community detection we make use of the geometric characterization of Section 4.2, suggesting that the vectors associated to different communities should be nearly orthogonal and vectors belonging to the same community should be nearly aligned. Given a choice of $d$ this suggests that we should use an angular $k$--means approach to form our clusters. Here, we present a measure to select the most appropriate value of $d$. 

Consider again the trivial example of $\ell$ disjoint communities each containing $z_\ell$ nodes. In this case, if we have an effective, normalized embedding we should have $$\langle X_i, X_j\rangle=\begin{cases} 1&\textrm{i and j belong to the same community}\\ 0 &\textrm{i and j belong to different communities}\end{cases}$$

Thus, the sum of intra--community dot products should be $\sum_{i=1}^\ell \binom{z_\ell}{2}$. Similarly, the sum of the inter--community dot products should be $0$. This suggests a stress function of the form $$s(d)=\sum_{i=1}^d \binom{z_i}{2}-\operatorname{s}_\textrm{intra} (d)+\operatorname{s}_\textrm{inter} (d)$$ where $\operatorname{s}_\textrm{intra} (d)$ is the sum of the dot products of all intra--community pairs and $\operatorname{s}_\textrm{inter}(d)$ is the sum of the dot products of all inter--community pairs. The dimension $d$, and its associated partition, that minimizes this value is then an appropriate candidate for partitioning the multi--network. 

As an example of this procedure, consider a stochastic block WRDPM  with intra--community parameter $1$ and inter--community parameter $.1$ for all communities as shown in Figure 9 (a). Figure 9 also displays the embeddings for $d=2$ and $d=3$  as well as the value of the stress function for $d\in\{2,3,4,5,6,7,8\}$. As expected the value is minimized for $d=3$.

\begin{figure}[!h]
\centering
\subfloat[Weighted Network]{\includegraphics[height=1.15in]{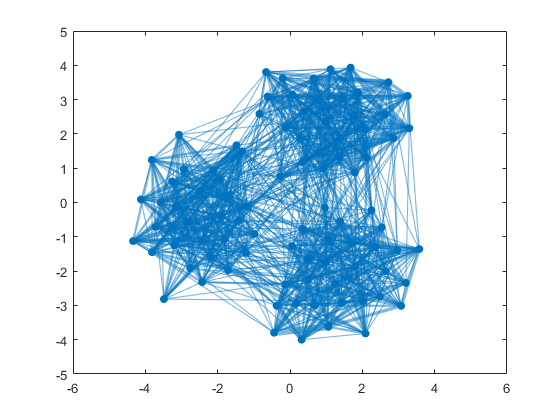}}  \ 
\subfloat[2-Embedding]{\includegraphics[height=1.15in]{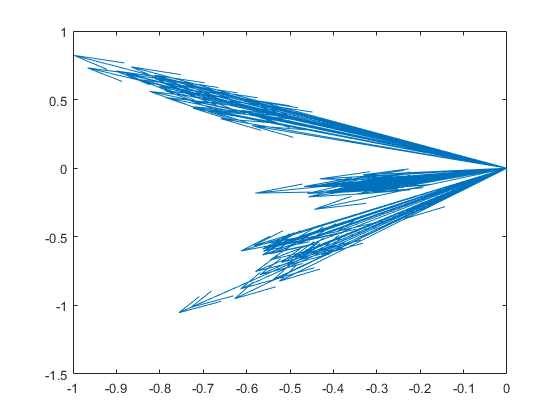}}\\
\subfloat[3-Embedding]{\includegraphics[height=1.45in]{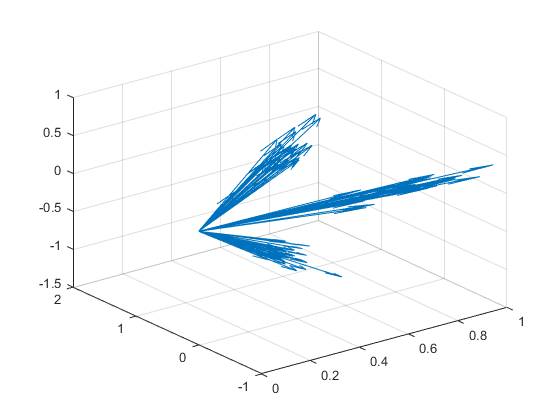}}  \ 
\subfloat[Stress Function]{\includegraphics[height=1.25in]{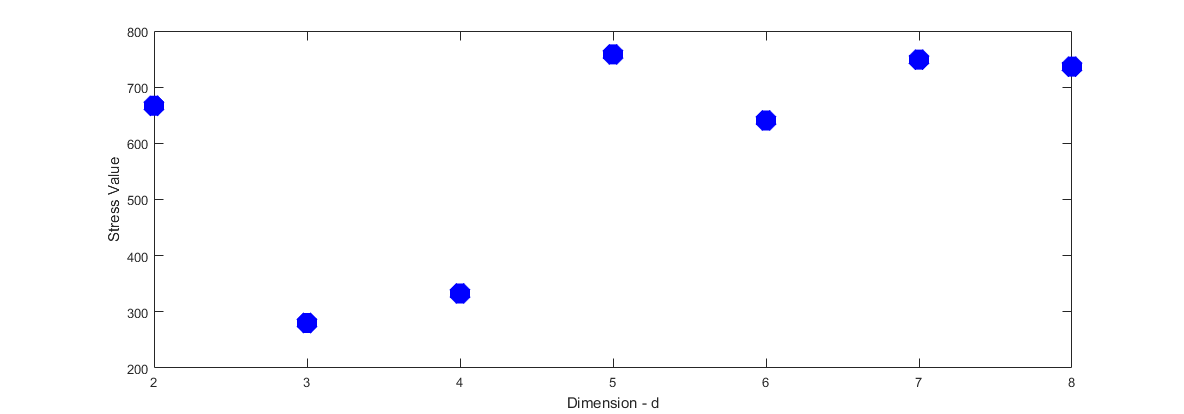}}
\caption{Comparison of WRDPN  embeddings of a weighted network (a) as the dimension of the embedding varies. As expected, the minimum value occurs at $d=3$, matching the correct structure.   }
\end{figure}

We note that it is not sufficient to consider either the inter or intra community dot products separately. For example, for the a higher dimensional embedding may split a community in such a way that the intra--community dot products are closer to $1$ for the subdivided community. This occurs with $d=4$ for the example in Figure 11 as the difference between the maximum possible sum of intra--community dot products and the realized sum is $\approx 8.22$ for $d=4$ and $\approx 9.97$ for $d=3$. However, this subdivision inflates the sum inter--community product for $d=4$ by adding the dot products, which are near $1$, from the original community that was subdivided.

Variations on this stress function may be better suited for particular applications. For example, communities of equal size can be prioritized by scaling the dot products by community size or comparing to $d\binom{\frac{n}{d}}{2}$ as a representation of the idealized situation. Alternatively, we could consider a function of the form  $$s_F(d,\lambda_1,\lambda_2)=\lambda_1(\sum_{i=1}^d \binom{z_i}{2}-\operatorname{s}_\textrm{intra} (d)+\operatorname{s}_\textrm{inter} (d))+\lambda_2||X^tX-A||_F$$ to prioritize better global embeddings by penalizing embeddings that are worse approxiomations of the original matrix.   

One situation where these variations can be helpful concerns networks with particularly well--defined communities. In this case a higher dimensional embedding $d>\ell$ may only separate a small collection of nodes or a single node from the original community structure and embed them nearly orthogonally to the remainder of the vectors. An example of this behavior is shown in Section 6.2. 

\subsection{WRDPM Null Models}

We mention that the approach of \cite{S_T} can be used to define a RDPM null model that does not appear elsewhere in the literature, by using the learned vectors as the $\{X_i\}$ in step (RDPM 4). 
This defines a probability distribution over $n$--node networks (represented by $n\times n$ symmetric binary adjacency matrices $A$) with $P(A|\{X_i\})= \prod_{\ell>j} \left[\left(\langle X_\ell, X_j\rangle\right)^{A_{\ell,j}}\left(1-\langle X_\ell, X_j\rangle\right)^{1-A_{\ell,j}}\right]$ \cite{C_M3}. In this latent space null model, the fixed $\{X_i\}$ and their respective pairwise dot products are the preserved attributes for this process.   The properties of the original network can then be compared to the expected properties of networks drawn from the distribution determined by the $\{X_i\}$. 

From the perspective of data analysis, we need not even begin with an explicit network structure. Instead, if we have pairwise similarity data for a collection of objects we can use the method of \cite{S_T} to find an embedding that models the given similarity scores by approximating the similarity values with the dot products of the $\{X_i\}$.  This allows for the construction  of an RDPM null model directly from the data  instead of from a derived network that may incorporate distortions \cite{T1,T2}. The conditions on permissible distributions $W$ for the standard RDPM are too restrictive to apply this technique in many situations, as the Bernoulli parameters for the edges require that the edge parameters must all be less than one. Depending on the selected distribution of edge weights, our generalization (cf. Section 2) encompasses a much larger class of distributions and hence allows for the construction of more accurate null models. 
\section{Applications}

We conclude by applying the method discussed in Section 4 to multi--networks derived from real--world data., a collaboration network of combinatorial geometers and voting data from the 112th U.S. Senate.  We compare the performance of the WRDPN  on multi--networks to the RDPM on binarized version of the same networks to demonstrate the benefits of the WRDPN . 
\subsection{Collaboration Networks}

Scientific collaboration networks are often studied as a proxy for the professional interaction networks of researchers (see \cite{N1} as an example). In the most common formulation of these networks, the nodes are scientists and two scientists are connected by an edge if they have written a paper together. However, these interactions also have a natural multigraph structure, where the number of edges between two scientists is computed as a (weighted) sum of the papers coauthored by them \cite{N2}.

 For these networks, we can interpret the two attributes of the vectors, direction and magnitude, in the context of our model. Two researchers are more likely to have a higher number of coauthored papers together if they share similar interests or connectivity patters, i.e., their vectors point in similar directions, or if one or both of them is particularly prolific, represented by large magnitude. Additionally, we expect researchers who collaborate across sub--topics to be assigned larger magnitude vectors.  We consider the large connected component of a collaboration network from the field of computational geometry \cite{B}, with 7,343 authors and 11,898 publications, where the edges are weighted by the number of co--publications. To compare to the RDPM model we also consider the unweighted underlying collaboration network. 

Using the methods of Section 4 we construct low--dimensional representations of the multi--network adjacency matrix for the giant component of the combinatorial geometry researcher data. Comparing the embeddings of the multigraph   to the embedding of the underlying unweighted simple graph Figure 12, which is much more uniformly distributed, shows that the clustering into ``nearly orthogonal'' components, centered on particularly prolific scientists/subfields, is much stronger in the multigraph setting than for the simple network. 
 
\begin{figure}[!h]
\centering
\subfloat[Weighted 2--Embedding]{\includegraphics[height=1.15in]{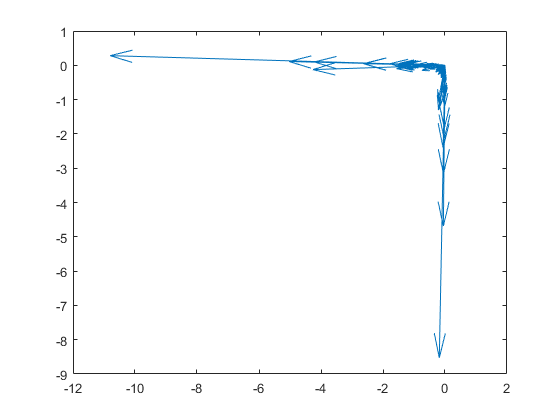}}  \ 
\subfloat[Unweighted 2--Embedding]{\includegraphics[height=1.15in]{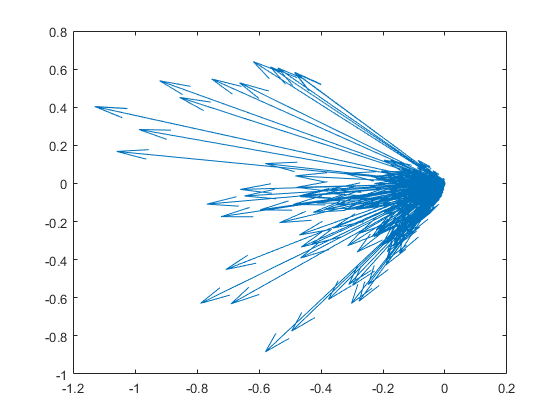}}\\
\subfloat[Weighted 3--Embedding]{\includegraphics[height=1.15in]{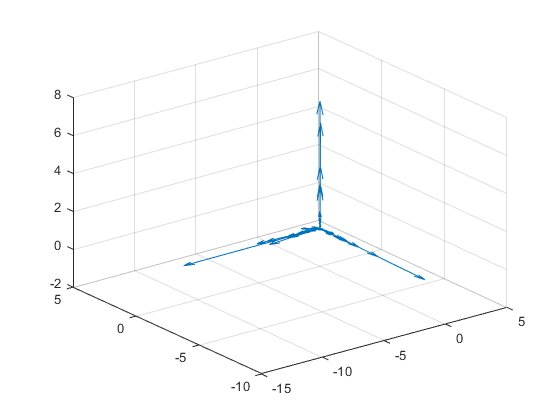}}  \ 
\subfloat[Unweighted 3--Embedding]{\includegraphics[height=1.15in]{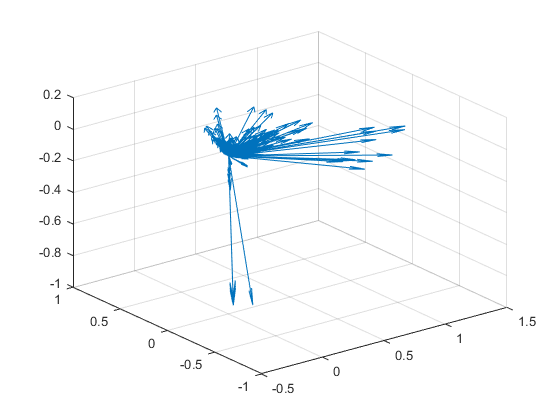}}
\caption{Comparison of WRDPM  embeddings to binarized embeddings of a weighted coauthorship network. The weighted embeddings partition the data much more cleanly. }
\end{figure}

We computed embeddings and their associated stress values for both networks with $d\in\{2,3,4,5,6,7,8,9,10\}$. The results are shown in Figure 13. Notice that the stress values associated with the unweighted embeddings are significantly higher than those of the weighted networks confirming the qualitative analysis in Figure 12. 

\begin{figure}[!h]

\centering
\includegraphics[height=1.15in]{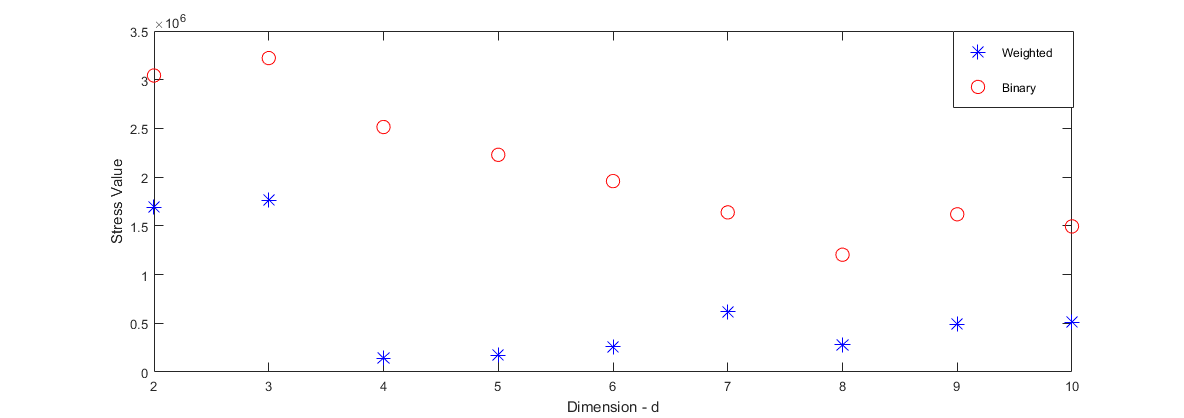}

\caption{Comparison of stress values for the computational geometry coauthorship network between the weighted and unweighted realizations. The weighted embedding significantly outperforms the binarized model. }
\end{figure}

\section{Voting Data}
We next consider roll call voting data from the 112th Senate \cite{SCP}. We construct a weighted adjacency matrix with $A_{i,j}$ representing the number of times that Senator $i$ cast the same vote as Senator $j$ on a piece of legislation. As with the coauthorship data, we used the methods of Section 4 to construct low dimensional approximations of the voting data. In this case, the binarized network is a complete graph which does not contain any useful clustering information. Standard approaches for deriving simple networks from weighted data, such as thresholding or bounded outdegrees, can encode improper structural distortions \cite{T1,T2}. Thus, the weighted model outperforms the binary version in this case as well. 

\begin{figure}[!h]
\centering
\subfloat[2--Embedding]{\includegraphics[height=1.15in]{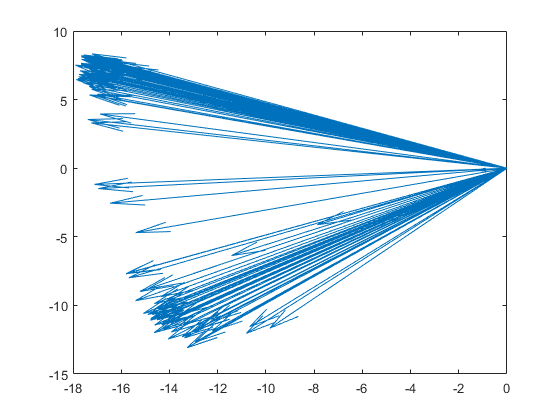}} \
\subfloat[3--Embedding]{\includegraphics[height=1.15in]{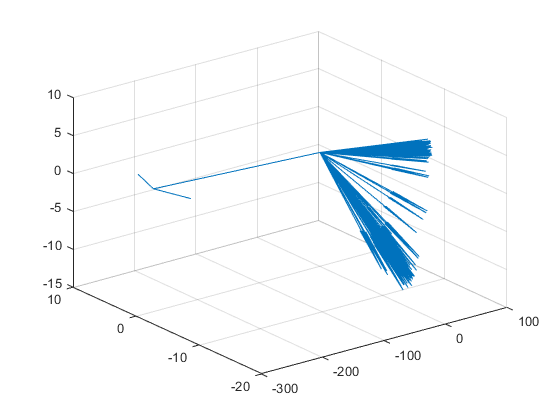}}
\caption{Embeddings of the 112 Senate roll call voting data. The two--dimensional embedding splits the Senate by party with the Republican Senators above and the Democratic Senators below. The three--dimensional embedding separates a single senator with an unusual voting pattern from the remaining 101 senators, suggesting that the partition achieved from the 2--embedding is robust.  }
\end{figure}

The two and three dimensional embeddings for the weighted networks are shown in Figure 14. The two dimensional embedding captures the party structure in the Senate, with the Republican Senators above the x-axis and the Democratic Senators below. The senators close to the center have reputations for moderation or cross party behavior such as Senator Scott Brown from Massachusetts. 

When the dimensionality of the embedding is increased, a single senator is separated from the two party structure. That individual is Senator John Ensign who resigned after only four months of the term. Thus, his voting record is quite distinct from the rest of the Senators. This pattern continues in higher dimensions with individual senators being separated from the party structure, suggesting that the division found in the two--dimensional case is close to optimal.

\section{Conclusion} 
We have introduced a generalization of the RDPM for weighted networks. As a generative model, the WRDPM contains several other commonly studied models as special cases and provides a geometric interpretation for these models.  Using the dot product to parametrize the network distinguishes the WRDPM (and the RDPM before it) from other latent space models where distance is the standard measure. 

The dot product embedding provides interpretatbility to the magnitude and direction of the vector associated to each node and allows for inference based on a factorization of a weighted adjacency matrix. This process leads to a dimension reduction procedure for weighted networks. Using the connection between the embeddings and community structure we constructed a stress function for dimension selection.

\end{document}